\pdfoutput=1
\documentclass[a4paper,USenglish,cleveref, autoref, thm-restate,authorcolumns]{lipics-v2019}
\nolinenumbers

\title{A Bridge between Polynomial Optimization and Games with Imperfect Recall}

 \author{Hugo Gimbert}    {LaBRI, CNRS, Université de Bordeaux, France} {hugo.gimbert@labri.fr}{}{}
 \author{Soumyajit Paul} {LaBRI,Université de Bordeaux, France} {soumyajit.paul@labri.fr}{}{}
 \author{B. Srivathsan}  {Chennai Mathematical Institute, India} {sri@cmi.ac.in}{}{}

\authorrunning{H.Gimbert, S.Paul and B.Srivathsan}



\hideLIPIcs 

\usepackage{tikz}
\usetikzlibrary{calc}
\usetikzlibrary{shapes,snakes}
\usepackage[utf8]{inputenc}

\usepackage[textsize=small]{todonotes}

\usepackage{multirow}

\newcommand{\incl}{\subseteq}

\renewcommand{\b}{\beta}
\renewcommand{\d}{\delta}
\renewcommand{\l}{\lambda}

\newcommand{\s}{\sigma}
\renewcommand{\r}{\rho}
\renewcommand{\t}{\tau}

\newcommand{\Cc}{\mathcal{C}}

\newcommand{\Oo}{\mathcal{O}}
\newcommand{\Tt}{\mathcal{T}}
\newcommand{\Uu}{\mathcal{U}}

\newcommand{\NP}{\operatorname{NP}}
\newcommand{\PSPACE}{\operatorname{PSPACE}}
\newcommand{\Ptime}{\operatorname{P}}
\newcommand{\sqsum}{\textsc{Square-Root-Sum}}

\newcommand{\Rat}{\mathbb{Q}}
\newcommand{\Real}{\mathbb{R}}

\newcommand{\xra}[1]{\xrightarrow{#1}}

\newcommand{\Max}{\mathsf{Max}}
\newcommand{\Min}{\mathsf{Min}}
\newcommand{\control}{\mathrm{Control}}
\newcommand{\payoff}{\mathrm{Payoff}}
\newcommand{\chance}{\mathsf{Chance}}
\newcommand{\moves}{\mathrm{Moves}}
\newcommand{\obs}{h}
\newcommand{\seq}{hist}

\newcommand{\pathto}{\operatorname{PathTo}}

\newcommand{\maxbeh}{\operatorname{Max_{beh}}}
\newcommand{\minbeh}{\operatorname{Min_{beh}}}
\newcommand{\maxminbeh}{\operatorname{MaxMin_{beh}}}
\newcommand{\maxpure}{\operatorname{Max_{pure}}}
\newcommand{\minpure}{\operatorname{Min_{pure}}}
\newcommand{\maxminpure}{\operatorname{MaxMin_{pure}}}
\newcommand{\minmaxpure}{\operatorname{MinMax_{pure}}}
\newcommand{\er}{\exists \mathbb{R}}
\newcommand{\fr}{\forall \mathbb{R}}
\newcommand{\efr}{\exists \forall \mathbb{R}}

\newcommand{\bx}{\bar{x}}
\newcommand{\cdeg}{\operatorname{c-deg}}

\begin{document}

\maketitle

\begin{abstract}   We provide several positive and negative complexity results for
  solving games with imperfect recall.  Using a one-to-one
  correspondence between these games on one side and multivariate
  polynomials on the other side, we show that solving games with
  imperfect recall is as hard as solving certain problems of the first
  order theory of reals. We establish square root sum hardness even
  for the specific class of A-loss games. On the positive side, we
  find restrictions on games and strategies motivated by Bridge
  bidding that give polynomial-time complexity.
 
\end{abstract}

\section{Introduction}
The complexity of games of finite duration and imperfect information
is a central problem in Artificial Intelligence.  In the particular
case of zero-sum two-player extensive form games with perfect recall,
the problem was notably shown to be solvable in
polynomial-time~\cite{KollerMegiddo::1992,vonStengel::1996}. The perfect recall
assumption, which states that players do not lose track of any
information they previously received, is mandatory for this
tractability result to hold: without this assumption, the problem was
shown to be NP-hard~\cite{KollerMegiddo::1992,Cermak::2018}.

The primary motivation for our work is to investigate the complexity
of the game of Bridge, a game between two teams of two players each:
North and South against West and East.  Bridge is a specific class of
multi-player games called \emph{team games}, where two teams of
players have opposite interests, players of the same team have the
same payoffs, but players cannot freely communicate, even inside the
same team (see e.g.~\cite{DBLP:conf/aaai/Celli018,NIPS2018_8172} for
more details).  Interestingly, dropping the perfect recall assumption
in zero-sum two player games is enough to encompass team games: the
lack of communication between the players about their private
information can be modeled with imperfect recall.  Another motivation
to study games with imperfect recall is that they may be used to
abstract large perfect recall games and obtain significant
computational improvements
empirically~\cite{Cermak::IJCAI::2017,PracticalUseImperfect}.

\medskip \newcommand{\FOTR}{FOT($\mathbb{R}$)} Our results exhibit
tight relations between the complexity of solving games with imperfect
recall and decision problems in the first-order theory of reals \FOTR.
A formula in \FOTR\ is a logical statement containing Boolean
connectives $\vee, \wedge, \neg$ and quantifiers $\exists, \forall$
over the signature $(0, 1, +, *, \le, <, =)$. We can consider it to be
a first order logic formula in which each atomic term is a polynomial
equation or inequation, for instance
$\exists x_1, x_2 \forall y (0 \le y \le 1) \rightarrow (4x_1y +
5x_2^2y + 3x_1^3x_2 > 4)$ (where we have used integers freely since
they can be eliminated without a significant blow-up in the size of
the formula~\cite{SchaeferNash}, and the implication operator $\to$ with
the usual meaning).  The complexity class $\er$ consists of those
problems which have a polynomial-time reduction to a sentence of the
form $\exists X \Phi(X)$ where $X$ is a tuple of real variables,
$\Phi(X)$ is a quantifier free formula in the theory of
reals. Similarly, the complexity classes $\fr$ and $\efr$ stand for
the problems that reduce to formulae of the form $\forall X \Phi(X)$
and $\exists X \forall Y \Phi(X, Y)$ where $X, Y$ are tuples of
variables. All these complexity classes $\er$, $\fr$ and $\efr$ are
known to be contained in $\PSPACE$~\cite{Canny:1988,Basu:2006}.
Complexity of games with respect to the $\er$ class has been studied
before in strategic form games, particularly for Nash equilibria
decision problems in $3$ player games~\cite{SchaeferNash, Garg:ICALP2015,
  BiloM16}.

Our paper provides several results about the complexity of extensive
form games with imperfect recall.  First, we show a one-to-one
correspondence between games of imperfect recall on one side and
multivariate polynomials on the other side and use it to establish
several results:
\begin{itemize}
\item In one-player games with imperfect recall, deciding whether the
  player has a behavioural strategy with positive payoff is
  $\er$-complete (Theorem~\ref{thm:one-player-absentminded}).  The
  same holds for the question of non-negative payoff.
\item In two-player games with imperfect-recall, the problem is in the
  fragment $\exists \forall \mathbb{R}$ of FOT($\mathbb{R}$) and it is
  both $\er$-hard and $\fr$-hard
  (Theorem~\ref{thm:two-player-absentminded}).  Even in the particular
  case where the players do not have absent-mindedness, this problem
  is $\sqsum$-hard (Theorem~\ref{thm:two-player-no-absentminded}).
\end{itemize}
A corollary is that the case where one of the two players has A-loss
recall and the other has perfect recall is $\sqsum$ hard, a question
which was left open in~\cite{Cermak::2018}.  While the above results
show that imperfect recall games are hard to solve, we also provide a
few tractability results.
\begin{itemize}
\item We capture the subclass of one-player perfect recall games with
  a class of \emph{perfect recall multivariate polynomials}. As a
  by-product we show that computing the optimum of such a polynomial
  can be done in polynomial-time, while it is NP-hard in general
  (Section~\ref{sec:polyn-optim}).  This also provides a heuristic to
  solve imperfect recall games in certain cases, by converting them to
  perfect recall games of the same size.
\item For one-player games where the player is bound to use
  deterministic strategies, the problem becomes polynomial-time when a
  parameter which we call the \emph{change degree} of the game is 
  constant (Theorem~\ref{thm:fixed-chance-degree}).

\item We provide a model for the bidding phase of the Bridge game, and
  exhibit a decision problem which can be solved in time polynomial in
  the size of the description (Lemma~\ref{lemma:nonoverbidding}).
  
\end{itemize}

\section{Games with imperfect information}
\label{sec:preliminaries}

This section introduces games with imperfect information.  These games
are played on finite trees by two players playing against each other
in order to optimize their payoff. The players are in perfect
competition: the game is zero-sum.  Nature can influence the game with
chance moves.  Players observe the game through information sets and
they are only partially informed about the moves of their adversary
and Nature.

\subsubsection*{Playing games on trees.}
For a set $S$, we write $\Delta(S)$ for a probability distribution
over $S$.

A finite \emph{directed tree} $\Tt$ is a tuple $(V, L, r, E)$ where
$V$ is a finite set of non-terminal nodes; $L$ is a non-empty finite
set of terminal nodes (also called \emph{leaves}) which are disjoint
from $V$; node $r \in V \cup L$ is called the \emph{root} and
$E \incl V \times (V \cup L)$ is the \emph{edge} relation. We write
$u \to v$ if $(u, v) \in E$. It is assumed that there is no edge
$u \to r$ incoming to $r$, and there is a unique path
$r \to v_1 \to \cdots \to v$ from the root to every $v \in V \cup
L$. We denote this path as $\pathto(v)$.

We consider games played between two players $\Max$ and $\Min$ along
with a special player $\chance$ to model random moves during the
game. We will denote $\Max$ as Player $1$ and $\Min$ as Player 2. An
extensive form \emph{perfect information game} is given by a tuple
$(\Tt, A, \control, \d, \Uu)$ where: $\Tt$ is a finite directed tree,
$A = A_1 \cup A_2$ is a set of actions for each player with
$A_1 \cap A_2 = \emptyset$, function
$\control: V \mapsto \{1,2\}\cup \{\chance\}$ associates each
non-terminal node to one of the players, $\d$ is a transition function
which we explain below, and $\Uu: T \mapsto \Rat$ associates a
rational number called the utility (or payoff) to each leaf. For
$i \in \{1, 2\}$, let $V_i$ denote the set of nodes controlled by
Player $i$, that is $\{v \in V~|~ \control(v) = i \} $ and let
$V_{\chance}$ denote the nodes controlled by $\chance$. We 
sometimes use the term \emph{control nodes} for nodes in
$V_1 \cup V_2$ and \emph{chance nodes} for nodes in $V_{\chance}$. The
transition function $\d$ associates to each edge $u \to v$ an action
in $A_i$ when $u \in V_i$, and a rational number when
$u \in V_{\chance}$ such that
$\sum_{v \text{ s.t. } u \to v} \d(u \to v) = 1$ (a probability
distribution over the edges of $u$). We assume that from control
nodes $u$, no two outgoing edges are labeled with the same action by
$\d$: that is $\d(u \to v_1) \neq \d(u \to v_2)$ when $v_1 \neq
v_2$. For a control node $u$, we write $\moves(u)$ for
$\{ a \in A_i~|~ a = \d(u \to v) \text{ for some } v \}$. Games $G_1$
and $G_2$ in Figure~\ref{fig:game-examples} without the blue dashed
lines are perfect information games which do not have $\chance$
nodes. Game $G_{-\sqrt{n}}$ of Figure~\ref{fig:sqrt-n} without the
dashed lines gives a perfect information game with $\Max$, $\Min$ and
$\chance$ where nodes of $\Max$, $\Min$ and $\chance$ are circles, squares and
triangles respectively.

An extensive form \emph{imperfect information game} is given by a
perfect information game as defined above along with two partition
functions $\obs_1: V_1 \mapsto \Oo_1$ and $\obs_2: V_2 \mapsto \Oo_2$
which respectively map $V_1$ and $V_2$ to a finite set of
\emph{signals} $\Oo_1$ and $\Oo_2$. The partition functions $\obs_i$
satisfy the following criterion: $\moves(u) = \moves(v)$ whenever
$\obs_i(u) = \obs_i(v)$. Each partition $h_i^{-1}(o)$ for
$o \in \Oo_i$ is called an \emph{information set} of Player
$i$. Intuitively, a player does not know her exact position $u$ in the
game, and instead receives the corresponding signal $\obs_i(u)$
whenever she arrives to $u$. Due to the restriction on moves, we can
define $\moves(o)$ for every $o \in \Oo_i$ to be equal to $\moves(u)$
for some $u \in h_i^{-1}(o)$. In Figure~\ref{fig:game-examples}, the
blue dashed lines denote the partition of $\Max$: in $G_1$, $\{r, u\}$
is one information set and in $G_2$, the information sets of $\Max$
are $\{u_1\}$, $\{u_2\}$ and $\{u_3, u_4\}$. $\Max$ has to play the
same moves at both $r$ and $u$ in $G_1$, and similarly at $u_3$ and
$u_4$ in $G_2$. Based on some structure of these information sets,
imperfect information games are further characterized into different
classes. We explain this next.

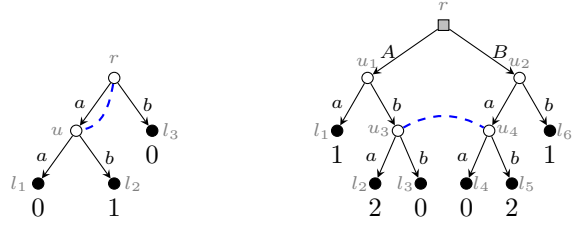
\begin{figure}[t]
\begin{center}
\tikzset{
triangle/.style = {regular polygon,regular polygon sides=3,draw,inner sep = 2},
circ/.style = {circle,draw,inner sep = 1.5},
term/.style = {circle,draw,inner sep = 1.5,fill=black},
sq/.style = {rectangle,fill=gray!50, draw, inner sep = 2}
}
\begin{tikzpicture}
\tikzstyle{level 1}=[level distance=7mm,sibling distance = 10mm]
\tikzstyle{level 2}=[level distance=7mm,sibling distance=10mm]
\tikzstyle{level 3}=[level distance=7mm,sibling distance=15mm]
\tikzstyle{level 4}=[level distance=7mm,sibling distance=8mm]


\begin{scope}[->, >=stealth]
\node (0) [circ] {}
child{
  node (1) [circ] {}
  child{
    node (3) [term, label=below:{$0$}] {}
    edge from parent node [left] {\scriptsize $a$}
  }
  child{
    node (4) [term,label=below:{$1$}] {}
    edge from parent node [right] {\scriptsize $b$}
  }
  edge from parent node [left] {\scriptsize $a$}
}
child{
  node (2) [term, label=below:{$0$}] {}
  edge from parent node [right] {\scriptsize $b$}
}
;
\end{scope}

\draw [dashed, thick, blue, in=10,out=-100] (0) to (1);
\node [gray] at (0,0.25) {\scriptsize $r$};
\node [gray] at (-0.75, -0.7) {\scriptsize $u$};
\node [gray] at (-1.25, -1.4) {\scriptsize $l_1$};
\node [gray] at (0.25, -1.4) {\scriptsize $l_2$};
\node [gray] at (0.75, -0.7) {\scriptsize $l_3$};
\end{tikzpicture}
\qquad \qquad
\begin{tikzpicture}
  \tikzstyle{level 1}=[level distance=7mm,sibling distance = 20mm]
\tikzstyle{level 2}=[level distance=7mm,sibling distance=8mm]
\tikzstyle{level 3}=[level distance=7mm,sibling distance=6mm]
\tikzstyle{level 4}=[level distance=7mm,sibling distance=5mm]


\begin{scope}[->, >=stealth]
\node (0) [sq] {}
child {
  node (1) [circ] {}
  child {
    node (3) [term, label=below:{$1$}] {}
    edge from parent node [left] {\scriptsize $a$}
  }
  child {
    node (4) [circ] {}
    child {
      node (7) [term, label=below:{$2$}] {}
      edge from parent node [left] {\scriptsize $a$}
    }
    child {
      node (8) [term, label=below:{$0$}] {}
      edge from parent node [right] {\scriptsize $b$}
      }
    edge from parent node [right] {\scriptsize $b$} 
  }
  edge from parent node [left] {\scriptsize $A$}
}
child {
  node (2) [circ] {}
   child {
     node (5) [circ] {}
     child {
      node (9) [term, label=below:{$0$}] {}
      edge from parent node [left] {\scriptsize $a$}
    }
    child {
      node (10) [term, label=below:{$2$}] {}
      edge from parent node [right] {\scriptsize $b$}
      }
    edge from parent node [left] {\scriptsize $a$}
  }
  child {
    node (6) [term, label=below:{$1$}] {}
    edge from parent node [right] {\scriptsize $b$} 
  }
  edge from parent node [right] {\scriptsize $B$}
}
;
\end{scope}

\draw [dashed, thick, blue, in=150,out=30] (4) to (5);
\node [gray] at (0,0.25) {\scriptsize $r$};
\node [gray] at (-1,-0.5) {\scriptsize $u_1$};
\node [gray] at (1, -0.5) {\scriptsize $u_2$};
\node [gray] at (-0.85, -1.4) {\scriptsize $u_3$};
\node [gray] at (0.85, -1.4) {\scriptsize $u_4$};
\node [gray] at (-1.6, -1.4) {\scriptsize $l_1$};
\node [gray] at (-1.1, -2.1) {\scriptsize $l_2$}; 
\node [gray] at (-0.5, -2.1) {\scriptsize $l_3$};
\node [gray] at (0.5, -2.1) {\scriptsize $l_4$};
\node [gray] at (1.1, -2.1) {\scriptsize $l_5$};
\node [gray] at (1.6, -1.4) {\scriptsize $l_6$};
\end{tikzpicture}
\end{center}
\caption{One player game $G_1$ on the left, and two player game $G_2$
  on the right}
\label{fig:game-examples}
\end{figure}

\subsubsection*{Histories.}
While playing, a player receives a sequence of signals, called the
\emph{history}, defined as follows.  For a vertex $v$ controlled by
player $i$, let

$\seq(v)$ be the sequence

of signals received and actions played by $i$ along $\pathto(v)$, the
path from the root to $v$.
For example in game $G_2$, $\seq(u_3) = \{u_1\}~b~\{u_3, u_4\}$ (for
convenience, we have denoted the signal corresponding to an
information set by the set itself). Note that the information set of a
vertex is the last signal of the sequence, thus if two vertices have
the same sequence, they are in the same information set. On the
other hand, the converse need not be true: two nodes in the same
information set could have different histories, for instance node
$u_4$ in $G_2$ has sequence $\{u_2\}~a~\{u_3, u_4\}$.

In such a case, what happens intuitively is that player $i$ does not
recall that she received the signals $\{u_1\}$ and $\{u_2\}$ and
played the actions $b$ and $a$.  This gives rise to various
definitions of \emph{recalls} for a player in the game.

\subsubsection*{Recalls.}
Player $i$ is said to have \emph{perfect recall} if she never forgets
any signals or actions, that is, for every $u, v \in V_i$, if
$\obs_i(u) = \obs_i(v)$ then $\seq(u) = \seq(v)$: every vertex in an
information set has the same history with respect to $i$. Otherwise
the player has \emph{imperfect recall}.

$\Max$ has imperfect recall in $G_1, G_2$ and $G_{-\sqrt{n}}$ whereas
$\Min$ has perfect recall in all of them (trivially, since there is
only one signal that she receives). Within imperfect recall we make
some distinctions.

Player $i$ is said to have \emph{absent-mindedness} if there are two
nodes $u, v \in V_i$ such that $u$ lies in the unique path from root
to $v$ and $h_i(u) = h_i(v)$ (player $i$ forgets not only her history,
but also the number of actions that she has played). $\Max$ has
absent-mindedness in $G_1$.

Player $i$ has \emph{A-loss} recall if she is not absent-minded, and
for every $u, v \in V_i$ with $h_i(u) = h_i(v)$ either
$\seq(u) = \seq(v)$ or $\seq(u)$ is of the form $\sigma a \sigma_1$
and $\seq(v)$ of the form $\sigma b \sigma_2$ where $\sigma$ is a
sequence ending with a signal and $a, b \in A_i$ with $a \neq b$
(player $i$ remembers the history upto a signal, after which she
forgets the action that she played). $\Max$ has A-loss in
$G_{-\sqrt{n}}$ since she forgets whether she played $a_0$ or $a_1$.
There are still cases where a player is not absent-minded, but not
A-loss recall either, 
for example when there
exists an information set containing $u, v$ whose histories differ
at a signal.
This happens when $i$ receives different signals due to the moves of
the other players (including player Chance), and later converges to
the same information set. In this document, we call such situations as
\emph{signal loss} for Player $i$. $\Max$ has signal loss in $G_2$
since at $\{u_3, u_4\}$ as she loses track between $\{u_1\}$ and
$\{u_2\}$.


\subsubsection*{Plays, strategies and maxmin value.} A \emph{play} is
a sequence of nodes and actions from the root to a leaf: for each leaf
$l$, the $\pathto(l)$ is a play. When the play ends at $l$, $\Min$
pays $\Uu(l)$ to $\Max$. The payoffs $\Uu(l)$ are the numbers below
the leaves in the running examples. $\Max$ wants to maximize the
expected payoff and $\Min$ wants to minimize it. In order to define
the expected payoff, we define the notion of \emph{strategies} for
each player.  A \emph{behavioural strategy} $\b$ for Player $i$ is a
function which maps each signal $o \in \Oo_i$ to $\Delta(\moves(o))$,
a probability distribution over its moves. For $a \in \moves(o)$, we
write $\b(o, a)$ for the value associated by $\b$ to the action $a$ at
information set $o$. For node $u$, we write $\b(u, a)$ for the
probability $\b(h_i(u), a)$.  A \emph{pure strategy} $\r$ is a special
behavioural strategy which maps each signal $o$ to a specific action
in $\moves(o)$.  We will denote the action associated at signal $o$ by
$\r(o)$, and for a node $u$ we will write $\r(u)$ for
$\r(h_i(u))$. For a node $u$ and an action $a$, we define
$\r(u, a) = 1$ if $\r(h_i(u)) = a$, and $\r(u, a) = 0$ otherwise. A
\emph{mixed strategy} is a distribution over pure strategies:
$\l_1 \r_1 + \l_2 \r_2 + \cdots + \l_k \r_k$ where each $\r_j$ is a
pure strategy, $0 \le \l_j \le 1$ and $\Sigma_j \l_j = 1$.

Consider a game $G$. Fixing behavioural strategies $\s$ for $\Max$ and
$\t$ for $\Min$ results in a game $G_{\s, \t}$ without control nodes:
every node behaves like a random node as every edge is labeled with a
real number denoting the probability of playing the edge. For a leaf
$t$, let $\Cc(t)$ denote the product of probabilities along the edges
controlled by $\chance$ in $\pathto(t)$. Let $\s(t)$ denote the
product of $\s(u, a)$ such that $u \in V_1$ and $u \xra{a} v$ is in
$\pathto(t)$. Similarly, let $\t(t)$ denote the product of the other
player's probabilities along $\pathto(t)$.  The payoff with these
strategies $\s$ and $\t$, denoted as $\payoff(G_{\s, \t})$ is then
given by:
$\sum_{t \in T} \Uu(t) \cdot \Cc(t) \cdot \s(t) \cdot \t(t)$.  This is
the ``expected'' amount that $\Min$ pays to $\Max$ when the strategies
are $\s$ and $\t$ for $\Max$ and $\Min$ respectively. We are
interested in computing $\max_{\s} \min_{\t} \payoff(G_{\s, \t})$. We
denote this value as $\maxminbeh(G)$ and call it the maxmin value
(over behavioural strategies). When $G$ is a one player game, the
corresponding values are denoted as $\maxbeh(G)$ or $\minbeh(G)$
depending on whether the single player is $\Max$ or $\Min$. We
correspondingly write $\maxminpure(G)$, $\maxpure(G)$ and
$\minpure(G)$ when we restrict the strategies $\s$ and $\t$ to be
pure. In the one player game $G_1$, $\maxpure(G_1)$ is $0$ since the
leaf $l_2$ is unreachable with pure strategies. Suppose $\Max$ plays
$a$ with probability $x$ and $b$ with $1 - x$, then $\maxbeh(G_1)$ is
given by $\max_{x \in [0, 1]} x (1-x)$. In $G_2$, a pure strategy for
$\Max$ can potentially lead to two leaves with payoffs either $1, 1$
or $1, 2$ or $2, 0$. Based on what $\Max$ chooses, $\Min$ can always
lead to the node with minimum among the two by appropriately choosing
the action at $r$. This gives $\maxminpure(G_2) = 1$. Observe that on
the other hand, $\minmaxpure(G_2) = 2$. It also turns out the
$\maxminbeh(G_2) = 1$, which can be shown by exploiting the symmetry
in the game.

 \section{Imperfect recall games}
\label{sec:imperf-recall-games}

In this section we investigate the complexity of imperfect
recall games and exhibit tight links 
with complexity classes arising out of the
first order theory of reals. Finding the maxmin value involves
computing a maxmin over polynomials where the variables are
partitioned between two players $\Max$ and $\Min$. It turns out that
imperfect recall games can capture polynomial manipulation entirely if
there is a single player. When there are two players, we 
 show that certain existential and
universal problems involving polynomials can be captured using
imperfect recall games.
 Previously, the only known lower bound was NP-hardness \cite{KollerMegiddo::1992}.
 We show that even the very specific case of 
two-player games
without absentmindedness is hard to solve:
 optimal values in such games can be
irrational and solving these games is $\sqsum$-hard.  A summary of
complexity results is given in Table
\ref{tab:complexity-results}.

\begin{figure}
  \begin{small}
  \begin{tabular}{|c||c|c|}
    \hline
    & \textit{No absentmindedness} & \textit{With absentmindedness} \\
    \hline
    \textit{One player} & $\NP$-complete & $\er$-complete ~(Theorem~\ref{thm:one-player-absentminded}) \\
    \hline
    \multirow{3}{*}{\textit{Two players}} & \multicolumn{2}{c|}{in $\efr$~ (Theorem~\ref{thm:two-player-absentminded})}  \\
    \cline{2-3}
    & $\sqsum$-hard & $\er$-hard and $\fr$-hard \\
    & (Theorem~\ref{thm:two-player-no-absentminded}) &
                                                     (Theorem~\ref{thm:two-player-absentminded}) \\
    \hline
  \end{tabular}
\end{small}
\caption{Complexity of imperfect recall games}
\label{tab:complexity-results}
\end{figure}


\subsection{One player}

We start with the hardness of games with a single player. The
important observation is that there is a tight connection between
multi-variate polynomials on one side and one-player games on the
other side.
 

\begin{lemma}
  \label{lem:polynomial-game}
  For every polynomial $F(x_1, \dots, x_n)$ over the reals, there
  exists a one player game $G_F$ with information sets
  $x_1, \dots, x_n$ such that the payoff of a behavioural strategy
  associating $d_i \in [0, 1]$ to $x_i$ is equal to
  $F(d_1, \dots, d_n)$.
\end{lemma}
\begin{proof}
  Suppose $F(x_1, \dots, x_n)$ has $k$ terms $\mu_1,...,\mu_k$. For
  each term $\mu_i$ in $F(x_1, \dots, x_n)$ we have a node $s_i$ in $G_F$ whose
  depth is equal to the total degree of $\mu_i$. From $s_i$ there is a
  path to a terminal node $t_i$ containing $d$ nodes for variable $x$,
  for each $x^{d}$ in $\mu_i$. Each of these nodes have two outgoing
  edges of which the edge not going to $t_i$ leads to a terminal node
  with utility $0$. In the terminal node $t_i$ the utility is equal to
  $k c_i$ where $c_i$ is the co-efficient of $\mu_i$.There is a root
  node belonging to $\chance$ which has transitions to each $s_i$ with
  probability $\frac{1}{k}$. All the other nodes belong to the single
  player.  All the nodes assigned due to a variable $x$ belong to one
  information set. The number of nodes is equal to sum of total
  degrees of each term. The payoffs are the same as the
  co-efficients. Hence the size of the game is polynomial in size of
  $F(x_1, \dots, x_n)$. Figure~\ref{fig:polynomial-game} shows an
  example (probability of taking $l$ from information set
  $\{u_1, u_2, u_3\}$ is $x$ and the probability of taking $l$ from
  $\{v_1, v_2, v_3\}$ is $y$).  Clearly the reduction from a
  polynomial to game is not unique.
\end{proof}

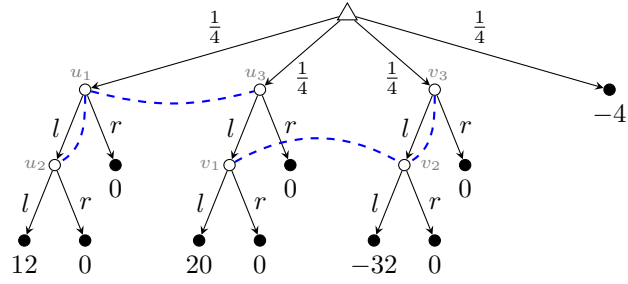
\begin{figure}
\begin{center}
\tikzset{
triangle/.style = {regular polygon,regular polygon sides=3,draw,inner sep = 1.5},
circ/.style = {circle,draw,inner sep = 1.5},
term/.style = {circle,draw,inner sep = 1.5,fill=black},
sq/.style = {rectangle,draw,inner sep = 2}
}
\begin{tikzpicture}
\tikzstyle{level 1}=[level distance=10mm,sibling distance=23mm]
\tikzstyle{level 2}=[level distance=10mm,sibling distance=8mm]
\tikzstyle{level 3}=[level distance=10mm,sibling distance=8mm]


\begin{scope}[->, >=stealth]
 \node(0)[triangle]{}
 child{
   node(1)[circ]{}
   child{
     node(6)[circ]{}
     child{
       node(12)[term,label=below:{$12$}]{}
       edge from parent node[left]{$l$}				
     }
     child{
       node(11)[term,label=below:{0}]{} 
       edge from parent node[right]{$r$}
     }
     edge from parent node[left]{$l$}				
   }
   child{
     node(5)[term,label=below:{0}]{} 
     edge from parent node[right]{$r$}
   }
   edge from parent node[above]{$\frac{1}{4}$}  
 }
 child{
   node(2)[circ]{}
   child{
     node(8)[circ]{}
     child{
       node(14)[term,label=below:{$20$}]{}
       edge from parent node[left]{$l$}				
     }
     child{
       node(13)[term,label=below:{0}]{}
       edge from parent node[right]{$r$}
     }
     edge from parent node[left]{$l$}
   }
   child{
     node(7)[term,label=below:{0}]{} 
     edge from parent node[right]{$r$}
   }
   edge from parent node[below]{$\frac{1}{4}$}
 }
 child{
   node(3)[circ]{}
   child{
     node(9)[circ]{}
     child{
       node(15)[term,label=below:{$-32$}]{} 
       edge from parent node[left]{$l$}
     }
     child{
       node(16)[term,label=below:{0}]{}
       edge from parent node[right]{$r$}				
     }
     edge from parent node[left]{$l$} 
   }
   child{
     node(10)[term,label=below:{0}]{}
     edge from parent node[right]{$r$}
   } 
   edge from parent node[below]{$\frac{1}{4}$}
 }
 child{
   node(4)[term,label=below:{$-4$}]{}
   edge from parent node[above]{$\frac{1}{4}$}
 }
 ;
\end{scope}
\draw [dashed,thick,blue,out=-15,in=-165] (1) to (2) 
  [dashed,blue,out=-90,in=30] (1) to (6);
 \draw [dashed,thick,blue,out=30,in=150](8) to (9)
 [dashed,blue,out=-90,in=30](3) to (9);

 \node [gray] at (-3.5, -0.8) {\scriptsize $u_1$};
 \node [gray] at (-4.1, -2) {\scriptsize $u_2$};
 \node [gray] at (-1.2, -0.8) {\scriptsize $u_3$};
 \node [gray] at (-1.8, -2) {\scriptsize $v_1$};
 \node [gray] at (1.1, -2) {\scriptsize $v_2$};
 \node [gray] at (1.2, -0.8) {\scriptsize $v_3$};
\end{tikzpicture}
\end{center}
\caption{One player game for the polynomial $3x^2 + 5xy - 8y^2 - 1$}
\label{fig:polynomial-game}
\end{figure}


The above lemma leads to the hardness of one player games.

\begin{lemma}
  \label{lem:one-player-hardness}
  The following two decision problems are $\er$-hard in one-player
  games with imperfect recall: (i) $\maxbeh \ge 0$ and (ii)
  $\maxbeh > 0$.
\end{lemma}
\begin{proof}
  \textit{(i)} The problem of checking if there exists a common root
  in $\Real^n$ for a system of quadratic equations $Q_i(X)$ is
  $\er$-complete~\cite{SchaeferNash}. This can be reduced to checking
  for a common root in $[0,1]^n$ using Lemma 3.9 of
  \cite{SchaeferRealizability}. We then reduce this problem to
  $\maxbeh \ge 0$. Note that $X$ is a solution to the system iff
  $- \sum_i Q_i(X)^2 \ge 0$. Using Lemma \ref{lem:polynomial-game} we
  construct a game $G_F$ with $F = - \sum_i Q_i(X)^2$. It then follows
  that the system has a common root iff $\maxbeh \ge 0$ in $G_F$.

  \textit{(ii)} We reduce $\maxbeh(G) \ge 0$ to $\maxbeh(G') > 0$ for
  some constructed game $G'$. Suppose that when $\maxbeh(G) < 0$, we can
  show $\maxbeh(G) < -\d$ for a constant $\d > 0$ that can be
  determined from $G$. With this claim, we have $\maxbeh(G) \ge 0$ iff
  $\maxbeh(G) + \d > 0$. We will then in polytime construct a game
  $G'$ whose optimal payoff is $\maxbeh(G) + \d$, which then proves
  the lemma. We will first prove the claim. The proof proceeds along
  the same lines as Theorem 4.1 in \cite{SchaeferNash}.

  Let $g(X)$ be the polynomial expressing the expected payoff in the
  game $G$ when the behavioural strategy is given by the variables
  $X$. Define two sets
  $S_1 := \{(z, X) \mid z = g(X), X \in [0,1]^n \}$ and
  $S_2 := \{ (0, X) \mid X \in [0,1]^n \}$. If $\maxbeh(G) < 0$, then
  $S_1$ and $S_2$ do not intersect. Since both $S_1, S_2$ are compact,
  this means there is a positive distance between them. Moreover,
  $S_1$ and $S_2$ are semi-algebraic sets (those that can expressed by
  a boolean quantifier free formula of the first order theory of
  reals). Corollary 3.8 of \cite{SchaeferNash} gives that this
  distance $ > 2^{-2^{L+5}}$ where $L$ is the complexity of the
  formulae expressing $S_1$ and $S_2$, which in our case is
  proportional to the size of the game. However, since $\d$ is doubly
  exponential, we cannot simply use it as a payoff to get
  $\maxbeh(G) + \d$.

  Define new variables $y_0, y_1, \dots, y_t$ for $t = L + 5$ and
  polynomials $F_i(y_0, \dots, y_t) := y_{i-1} - y^2_i$ for
  $i \in \{1, \dots, t-1\}$ and
  $F_t(y_0, \dots, y_t) := y_t - \frac{1}{2}$. The only common root of
  this system of polynomials $F_i$ gives $y_0 = 2^{-2^{t}} = \d$. Let
  $P := - \sum_i F_i^2(y_0, \dots, y_t)$ and let $G_P$ be the
  corresponding game as in Lemma~\ref{lem:polynomial-game}. Construct
  a new game $G'$ as follows. Its root node is a $\chance$ node with
  edges to three children each with probability $\frac{1}{3}$. To the
  first child, we attach the game $G$, and to the second child, the
  game $G_P$. The third child is node which is controlled by $\Max$
  and belongs to the information set for variable $y_0$. It has two
  leaves as children, the left with payoff $0$ and the right with
  payoff $1$. Observe that the optimal payoff for $\max$ in $G'$ is
  $\frac{1}{3}(\maxbeh(G) + \d)$. From the discussion in the first
  paragraph of this proof, we have $\maxbeh(G) \ge 0$ iff
  $\maxbeh(G') > 0$.
\end{proof}

The previous lemma shows that the game problem is $\er$-hard. Inclusion in
$\er$ is straightforward since the payoff is given by a polynomial
over variables representing the value of a behavioural strategy at
each information set. For example, for the game $G_1$ of
Figure~\ref{fig:game-examples}, deciding $\maxbeh(G_1) \ge 0$ is
equivalent to checking
$\exists x ( 0 \le x \le 1 \land x(1-x) \ge 0)$. We thus get the
following theorem.

\begin{theorem}
  \label{thm:one-player-absentminded}
  For one player games with imperfect recall, deciding $\maxbeh \ge 0$
  is $\er$-complete. Deciding $\maxbeh > 0$ is also $\er$-complete.
\end{theorem}

\subsection{Two players }

We now consider the case with two players. Analogous to the one player
situation, now $\maxminbeh(G) \ge 0$ can be expressed as a formula in
$\efr$. For instance, consider the game $G_2$ of
Figure~\ref{fig:game-examples}. Let $x, y, z, w$ be the
probability of taking the left action in $u_1,u_2, \{u_3, u_4\}$ and $r$
respectively. Deciding $\maxminbeh(G_2) \ge 0$ is equivalent to the
formula
$\exists x, y, z \forall w (0 \le w \le 1 \to (wx + 2w(1-x)z +
2(1-w)y(1-z) + (1-w)(1-y) \ge 0))$. This gives the upper bound on the
complexity as $\efr$. Hardness is established below.

\begin{theorem}
  \label{thm:two-player-absentminded}
  Deciding $\maxminbeh(G) \ge 0$ is in $\exists \forall
  \mathbb{R}$. It is both $\er$-hard and $\fr$-hard.
\end{theorem}
\begin{proof}
  Inclusion in $\efr$ follows from the discussion above.  For the
  hardness, we make use of Lemma~\ref{lem:one-player-hardness}. Note
  that when there is a single player $\Max$, $\maxbeh(G) \ge 0$ is the
  same as $\maxminbeh(G)$ $\ge 0$. As the former is $\er$-hard, we get
  the latter to be $\er$-hard. Now we consider the
  $\fr$-hardness. Since $\maxbeh(G) > 0$ is also $\er$-hard, the
  complement problem $\maxbeh(G) \le 0$ is $\fr$-hard. Hence the
  symmetric problem $\minbeh(G) \ge 0$ is $\fr$-hard. This is
  $\maxminbeh(G) \ge 0$ when there is a single player $\Min$, whence
  $\maxminbeh(G) \ge 0$ is $\fr$-hard.
\end{proof}

In these hardness results, we crucially use the squaring
operation. Hence the resulting games need to have
absentmindedness. Games without absentmindedness result in multilinear
polynomials. The hardness here comes due to irrational
numbers. Examples were already known where maxmin behavioural
strategies required irrational numbers \cite{KollerMegiddo::1992} but
the maxmin payoffs were still rational. We generate a class of games
where the maxmin payoffs are irrational as well. The next lemma lays
the foundation for Theorem~\ref{thm:two-player-no-absentminded}
showing square root sum hardness for this problem. The $\sqsum$
problem is to decide if $\sum_{i=1}^{m} \sqrt{a_i} \le p$ for given
positive integers $a_1, \dots, a_m, p$. This problem was first
proposed in \cite{Garey:1976}, whose complexity was left as an open
problem. The notion of $\sqsum$-hardness was put forward in
\cite{Etessami:2005} and has also been studied with respect to
complexity of minmax computation \cite{Hansen:2010} and game
equilibrium computations \cite{DBLP:journals/siamcomp/EtessamiY10}. In
\cite{Etessami:2005, Hansen:2010} the version discussed was to decide
if $\sum_{i=1}^{m} \sqrt{a_i} \ge p$. But our version is
computationally same since the equality version is decidable in
$\Ptime$ \cite{Blomer:1991}. The $\sqsum$ problem is not known to be
in $\NP$. It is known to lie in the Counting Hierarchy
\cite{Allender:2009} which is in $\PSPACE$. 


When $\Max$ has A-loss recall and $\Min$ has perfect recall, deciding
maxmin over behavioural strategies is $\NP$-hard
\cite{Cermak::2018}. The question of whether it is $\sqsum$-hard was
posed in \cite{Cermak::2018}. We settle this problem by showing that
even with this restriction it is $\sqsum$-hard. 

\begin{lemma}
  \label{lem:sqsum-gadget}
  For each $n \ge 0$, there is a two-player game $G_{-\sqrt{n}}$
  without absentmindedness such that
  $\maxminbeh(G_{-\sqrt{n}}) = -\sqrt{n}$.
\end{lemma}
\begin{proof}
  First we construct a game $G_1$ whose maxmin value is
  $\frac{n(n+1 - 2\sqrt{n})}{(n-1)^2}$ from which we get a game $G_2$
  with maxmin value $n+1 - 2\sqrt{n}$ by multiplying the payoffs of
  $G_1$ with $\frac{(n-1)^2}{n}$. Then we take a trivial game $G_3$
  with maxmin value $-(n+1)$ and finally construct $G_{-\sqrt{n}}$ by
  taking a root vertex $r$ as chance node and transitions with $1/2$
  probability from $r$ to $G_2$ and $G_3$.

  We now describe the game $G_1$. The game tree has 7 internal nodes
  and 16 leaf nodes with payoffs. At the root node $s_{\epsilon}$,
  there are 2 actions $a_0$ and $a_1$, playing which the game moves to
  $s_0$ or $s_1$. Then again at $s_i$ the action $b_0$ and $b_1$ are
  available playing which the game can go to $s_{0,0},s_{0,1},s_{1,0}$
  or $s_{1,1}$. And finally again playing action $c_0$ or $c_1$ the
  game can go to the leaf states
  $\{ t_{i,j,k} \mid i,j,k \in \{0,1\} \}$.  The node $s_{\epsilon}$
  is in one information set $I_1$ and belongs to $\Max$. The nodes
  $s_0$ and $s_1$ are in one information set $I_2$ and also belong to
  $\Max$. Nodes $s_{0,0},s_{0,1},s_{1,0}$ and $s_{1,1}$ are in the
  same information set $J$ and belong to $\Min$. The payoff at
  $t_{0,0,0}$ is $n$ and the payoff at $t_{1,1,1}$ is 1. Everywhere
  else the payoff is $0$.  \begin{figure}[t]
\begin{center}
\tikzset{
triangle/.style = {regular polygon,regular polygon sides=3,draw,inner sep = 1.5},
circ/.style = {circle,draw,inner sep = 1.5},
term/.style = {circle,draw,inner sep = 1.5,fill=black},
sq/.style = {rectangle,draw,inner sep = 2}
}
\begin{tikzpicture}
\tikzstyle{level 1}=[level distance=10mm,sibling distance = 25mm]
\tikzstyle{level 2}=[level distance=10mm,sibling distance=30mm]
\tikzstyle{level 3}=[level distance=10mm,sibling distance=15mm]
\tikzstyle{level 4}=[level distance=10mm,sibling distance=8mm]

\begin{scope}[->,>=stealth]

\node(-1)[triangle]{}
child{
node(0)[circ]{}	
	child{
	node(1)[circ]{}	
		child{
		node(3)[sq]{}
			child{
			node(7)[term,label=below:{$(n-1)^2$}]{}
			edge from parent node[left]{$c_0$}
			}
			child{
			node(8)[term,label=below:{0}]{}
			edge from parent node[right]{$c_1$}	
			}
		edge from parent node[left]{$b_0$}
		}
		child{
		node(4)[sq]{}
			child{
			node(9)[term,label=below:{0}]{}
			edge from parent node[left]{$c_0$}
			}
			child{
			node(10)[term,label=below:{0}]{}
			edge from parent node[right]{$c_1$}	
			}
		edge from parent node[right]{$b_1$}	
		}
	edge from parent node[above]{$a_0$}
	}
	child{
	node(2)[circ]{}
		child{
		node(5)[sq]{}
			child{
			node(11)[term,label=below:{0}]{}
			edge from parent node[left]{$c_0$}
			}
			child{
			node(12)[term,label=below:{0}]{}
			edge from parent node[right]{$c_1$}	
			}
		edge from parent node[left]{$b_0$}
		}
		child{
		node(6)[sq]{}	
			child{
			node(13)[term,label=below:{0}]{}
			edge from parent node[left]{$c_0$}
			}
			child{
			node(15)[term,label=below:{$\frac{(n-1)^2}{n}$}]{}
			edge from parent node[right]{$c_1$}	
			}
		edge from parent node[right]{$b_1$}
		}
	edge from parent node[above]{$a_1$}
	}
edge from parent node[above]{$\frac{1}{2}$}	
}
child{
node(16)[term,label=below:{$-(n+1)$}]{}
edge from parent node[above]{$\frac{1}{2}$}	
}	
;
\end{scope}

\draw [dashed,thick,blue,out=-15,in=-170] (1) to (2);
\draw [dashed,thick,red,out=15,in=165] (3) to (4)
[out=15,in=165] (4) to (5)
[out=15,in=165] (5) to (6)
;
\end{tikzpicture}

\end{center}
\caption{Game $G_{-\sqrt{n}}$}
\label{fig:sqrt-n}
\end{figure}

  Figure~\ref{fig:sqrt-n} depicts the game $G_{-\sqrt{n}}$ and the
  left subtree from chance node is $G_1$ after scaling the payoffs by
  $\frac{(n-1)^2}{n}$.  We wish to compute the maxmin value obtained
  when both the players play behavioural strategies.  Assigning variables
  $x,y,z$ for information sets $I_1,I_2,J$ respectively, the maxmin
  value is given by the expression
  \[
    \max_{x,y \in [0,1]} \min_{z \in [0,1]} nxyz + (1-x)(1-y)(1-z)
  \]
  which in this case is equivalent to
  \[
    \max_{x,y \in [0,1]} \min ( nxy , (1-x)(1-y) )
  \]
  since the best response of $\Min$ is given by a pure strategy when
  $\Min$ has no absentmindedness. It turns out this value is achieved
  when $nxy = (1-x)(1-y)$.
  We use this to get rid of $y$ and reduce to:
  \[
    \max_{x \in [0,1]} \frac{nx(1-x)}{1+(n-1)x}
  \]
  Calculating this we see that the maximum in $[0,1]$ is achieved at
  $x = \frac{\sqrt{n} -1}{n-1}$. After evaluation we get
  $\maxminbeh(G_1) = \frac{n(n+1 - 2\sqrt{n})}{(n-1)^2}$ as intended,
  at $ x = y = \frac{\sqrt{n} -1}{n-1}$.
\end{proof}

\begin{theorem}
  \label{thm:two-player-no-absentminded}
  Deciding $\maxminbeh \ge 0$ is $\sqsum$-hard in imperfect recall
  games without absentmindedness.
\end{theorem}
\begin{proof}
  From the positive integers $a_1,...,a_m$ and $p$ which are the
  inputs to the $\sqsum$ problem, we construct the following game
  $\hat{G}$. At the root there is a chance node $\hat{r}$. From
  $\hat{r}$ there is a transition with probability $\frac{1}{m+1}$ to
  each of the games $G_{-\sqrt{a_i}}$ (as constructed in
  Lemma~\ref{lem:sqsum-gadget}) and also a trivial game with payoff
  $p$. Now $\Max$ can guarantee a payoff $0$ in $\hat{G}$ iff
  $\sum_{i=1}^m \sqrt{a_i} \leq p$.
\end{proof}
 
In the proof above since in each of $G_{-\sqrt{n}}$, $\Max$ has A-loss
recall and $\Min$ has perfect recall, the same holds in $\hat{G}$.
Hence it is $\sqsum$-hard to decide the problem even when $\Max$ has
A-loss recall and $\Min$ has perfect recall.

 \section{Polynomial optimization}
\label{sec:polyn-optim}
In Section~\ref{sec:imperf-recall-games} we have seen that
manipulating polynomials can be seen as solving one-player imperfect
recall games (Lemma~\ref{lem:polynomial-game} and
Figure~\ref{fig:polynomial-game}). In particular, optimizing a
polynomial with $n$ variables over the domain $[0, 1]^n$ (the unit
hypercube) can be viewed as finding the optimal payoff in the
equivalent game. On the games side, we know that games with perfect
recall can be solved in polynomial time~\cite{KollerMegiddo::1992,
  vonStengel::1996}. We ask the natural question on the polynomials
side: what is the notion of perfect recall in polynomials? Do perfect
recall polynomials correspond to perfect recall games? We answer this
question in this section.

Consider a set $X$ of real variables. For a variable $x \in X$, we
write $\bx = 1 - x$ and call it the \emph{complement} of $x$. Let
$\bar{X} = \{ \bx \mid x \in X \}$ be the set of complements.  We
consider polynomials with integer coefficients having terms over
$X \cup \bar{X}$. Among such polynomials, we restrict our attention to
multilinear polynomials: each variable appearing in a term has degree
$1$ and no term contains a variable and its complement. Let $M(X)$ be
the set of such polynomials. For example
$3xyz - 5 \bar{x} \bar{y} z + 9 \bar{z} \in M(\{x, y, z\})$ whereas
$4x\bar{x} \not \in M(\{x\})$ and $4x^2 \notin M(\{x\})$. 

For $f, g \in M(X)$ we write $f \equiv g$ if eliminating the negations
from $f$ and $g$ gives the same full expansion. For example,
$y - yx \equiv y\bar{x}$ and $y\bar{x} + x \equiv y + x\bar{y}$. By
definition, the full expansion $f'$ of a polynomial $f$ satisfies
$f \equiv f'$. Also note that $\equiv$ is an equivalence relation.

We are interested in the problem of optimizing a polynomial
$f \in M(X)$ over the unit hypercube $[0, 1]^{|X|}$. The important
property is that the optimum occurs at a vertex. This corresponds to
saying that in a one-player imperfect recall game without
absentmindedness, the optimum is attained at a pure strategy (which is
shown by first proving that every behavioural strategy has an
equivalent mixed strategy and hence there is at least one pure
strategy with a greater value).  Due to this property, the decision
problem is in $\NP$. Hardness in $\NP$ follows from Corollary 2.8 of
\cite{KollerMegiddo::1992}.

\begin{theorem}[\cite{KollerMegiddo::1992}]
  The optimum of a polynomial in $M(X)$ over the unit hypercube
  $[0, 1]^{|X|}$ occurs at a vertex. Deciding if the maximum is
  greater than or equal to a rational is $\NP$-complete.
\end{theorem}

Our goal is to characterize a subclass of polynomials which coincide
with the notion of perfect recall in
games. 
For this we assume that games have exactly two actions from each
information set (any game can be converted to this form in
polynomial-time). The polynomials arising out of such games will come
from $M(X)$ where going left on information set $x$ gives terms with
variable $x$ and going right gives terms with $\bar{x}$. When the game
has perfect recall, every node in the information set of $x$ has the
same history: hence if some node in an information set $y$ is reached
by playing left from an ancestor $x$, every node in $y$ will have this
ancestor and action in the history. This implies that every term
involving $y$ will have $x$. If the action at $x$ was to go right to
come to $y$, then every term with $y$ will have $\bar{x}$. This
translates to a decomposition of polynomials in a specific form.

A polynomial $g$ given by $x f_0(X_0) + \bar{x} f_1(X_1) + f_2(X_2)$
is an \emph{$x$-decomposition} of a polynomial $f$ if
$x \notin X_0 \cup X_1 \cup X_2$ and expanding all complements in $g$
and $f$ result in the same complement-free polynomial. The
decomposition $g$ is said to be \emph{disconnected} if $X_0, X_1, X_2$
are pairwise disjoint. For example $g := xyz + 4\bar{x}y + 5\bar{w}$
is an $x$-decomposition of $xyz + 4y - 4xy + 5 - 5w$ which is not
disconnected due to variable $y$.  Using these notions, we now define
perfect recall polynomials in an inductive manner.

\begin{definition}[Perfect recall polynomials]
  Every polynomial over a single variable has perfect recall. A
  polynomial $f$ with variable set $X$ has perfect recall if there
  exists an $x \in X$ and an $x$-decomposition
  $x f_0(X_0) + \bar{x} f_1(X_1) + f_2(X_2)$ of $f$ such that (1) it
  is disconnected and (2) each $f_i(X_i)$ has perfect recall.
\end{definition}

This definition helps us to inductively generate a perfect recall game
out of a perfect recall polynomial and vice-versa, giving us the following
theorem.

\begin{theorem}
  \label{thm:perfect-recall-poly-game}
  A polynomial $f$ in $M(X)$ has perfect recall iff there is a
  one-player perfect recall game whose payoff is given by $f$. This
  transformation from perfect recall polynomial to one-player perfect
  recall game can be computed in polynomial time.
\end{theorem}

We prove both directions of the above theorem separately in the following lemmas.
The proof below showcases a stronger result that from a perfect recall polynomial, we can in fact construct a
perfect information game. 

\begin{lemma}
  For every perfect recall polynomial $f$, there is a perfect
  information game with payoff given by $f$.
\end{lemma}
\begin{proof}
  We construct the game inductively. For single variable polynomials
  $c_0 x + c_1 \bar{x}$, the game has a single non-terminal node with
  two leaves as children. The left leaf has payoff $c_0$ and the right
  has payoff $c_1$. The behavioural strategy at this single node is
  given by $x$ to the left node and $\bar{x}$ to the right node and
  hence the payoff is given by $c_0 x + c_1 \bar{x}$. Now consider a
  perfect recall polynomial with multiple variables. Consider the
  $x$-decomposition $x f_0(X_0) + \bar{x} f_1(X_1) + f_2(X_2)$ which
  witnesses the perfect recall. Each $X_i$ has fewer variables since
  $x$ is not present. By induction, there are perfect recall games
  $G_0, G_1, G_2$ whose payoffs are given by $f_0, f_1, f_2$
  respectively. Construct game $G$ with the root being a $\chance$
  node with two transitions each with probability $\frac{1}{2}$. To
  the right child attach the game $G_2$. The left child is a control
  node with left child being game $G_0$ and the right child being
  $G_1$. This node corresponds to variable $x$. Finally multiply all
  payoffs at the leaves with $2$. The payoff of this game is given by
  $x f_0(X_0) + \bar{x} f_1(X_1) + f_2(X_2)$. Since the decomposition
  is disconnected, the constructed is also perfect recall.  This
  construction gives us a perfect information game.
\end{proof}

\begin{lemma}
  The payoff of a perfect recall game is given by a perfect recall
  polynomial.
\end{lemma}
\begin{proof}
  Once again, proof proceeds by induction.  Every game with a single
  information set is clearly perfect recall and the payoff polynomial
  is perfect recall by definition. Pick a game $G$ with multiple
  information sets. We need to consider two cases depending on the
  root node.

  Suppose the root $r$ of $G$ is a control node with information set
  $x$. Since $G$ is perfect recall, no other node is present in this
  information set $x$. Let $G_0, G_1$ be the left and right subtree of
  $r$. Again, as $G$ has perfect recall, no information set straddles
  across the two subtrees. Hence the payoff of $G$ can be written as a
  disconnected $x$-decomposition $x f_0(X_0) + \bar{x} f_1(X_1)$ where
  $f_0, f_1$ are the payoffs of $G_0$ and $G_1$
  respectively. Moreover, the games $G_0$ and $G_1$ have perfect
  recall. By induction, the payoffs $f_0, f_1$ are perfect recall
  polynomials.

  Suppose the root $r$ belongs to $\chance$. Walking along some path
  from the root, we will hit the first node that is controlled by the
  player. Let $x$ be the information set for this node. As the player
  has perfect recall, for every node in $x$ the path from the root to
  it contains only $\chance$ nodes. Let $L_0, L_1$ be the set of
  leaves that are reached by taking respectively the left or right
  action from a node in $x$. Let $L_2$ be all the other leaves in
  $G$. The payoff of $G$ can be written as
  $x f_0(X_0) + \bar{x} f_1(X_1) + f_2(X_2)$ where $x f_0(X_0)$ gives
  the contribution of $L_0$, $\bar{x} f_1(X_1)$ gives that of $L_1$
  and $f_2(X_2)$ gives the payoff from $L_2$. This polynomial is an
  $x$-decomposition which is disconnected since $G$ has perfect
  recall. It remains to show that $f_0, f_1, f_2$ are perfect recall
  polynomials. For this we show that there are perfect recall games
  $G_0, G_1, G_2$ with fewer variables that yield $f_0, f_1,
  f_2$. Induction hypothesis then tells that they are perfect recall
  polynomials. Game $G_0$ is as follows: root node belongs to
  $\chance$; add a transition from root to all left subtrees of nodes
  in $x$; if there are $m$ such subtrees then each transition has
  probability $\frac{1}{m+1}$; finally multiply all payoffs by
  $m+1$. Game $G_1$ is similarly constructed by taking right
  subtrees. Game $G_2$ is obtained from $G$ by replacing subtrees
  starting from $x$ by leaves with payoff $0$. Each of these
  constructed games preserves perfect recall.
\end{proof}

Theorem~\ref{thm:perfect-recall-poly-game} allows to optimize perfect recall polynomials in
polynomial-time by converting them to a game. However, for this to be
algorithmically useful, we also need an efficient procedure to check
if a given polynomial has perfect recall. For games, checking perfect
recall is an immediate syntactic check. For polynomials, it is not
direct. We establish in this section that checking if a polynomial has
perfect recall can also be done in polynomial-time. The crucial
observation that helps to get this is the next proposition.

\begin{proposition}
  \label{prop:perfect-poly-disconnected-decomp}
  If a polynomial $f$ has perfect recall, then in every disconnected
  $x$-decomposition $x f_0(X_0) + \bar{x} f_1(X_1) + f_2(X_2)$ of $f$,
  the polynomials $f_0(X_0)$, $f_1(X_1)$ and $f_2(X_2)$ have perfect
  recall.
\end{proposition}

 We will first prove the above proposition
through some intermediate observations. The lemma below follows by
definition of perfect recall polynomials and the relation $\equiv$
between polynomials.

\begin{lemma}
  A polynomial $f$ has perfect recall iff its full expansion has
  perfect recall.
\end{lemma}

\begin{corollary}
  \label{cor:equiv-polynomials-perfect}
  Let $f, g$ be polynomials such that $f \equiv g$. Then $f$ has
  perfect recall iff $g$ has perfect recall.
\end{corollary}

\begin{lemma}
  \label{lem:relating-disconnected-decomp}
  Let $x f_0(X_0) + \bar{x} f_1(X_1) + f_2(X_2)$ and
  $y g_0(Y_0) + \bar{y} g_1(Y_1) + g_2(Y_2)$ be two disconnected
  decompositions of $f$. Then:
  \begin{enumerate}
  \item either $xf_0 \equiv yg_0$, $\bar{x}f_1 \equiv \bar{y}g_1$ and
    $f_2 \equiv g_2$,
  \item or $xf_0 \equiv \bar{y} g_1$, $\bar{x} f_1 \equiv y g_0$ and
    $f_2 \equiv g_2$,
  \item or $xf_0 + \bar{x} f_1 \equiv g_2$ and
    $f_2 \equiv y g_0 + \bar{y} g_1$
  \end{enumerate}
\end{lemma}
\begin{proof} When $x = y$, we can show the first statement of the
  lemma. When $x \neq y$, we need to consider the following cases: (a)
  $y \in X_0$ and $x \in Y_0$, (b) $y \in X_1$ and $x \in Y_0$ and (c)
  $y \in X_2$ and $x \in Y_2$. The other cases are either symmetric or
  impossible. Cases (a), (b), (c) entail the first, second or third
  statements of the lemma respectively. Proof proceeds by routine
  analysis of the terms in the full expansion of $f$.
\end{proof}

\paragraph*{Proof of
  Proposition~\ref{prop:perfect-poly-disconnected-decomp}.}
Proof proceeds by induction on the number of variables. When there is
a single variable, the proposition is trivially true. Consider
polynomial $f$ over multiple variables. Since it has perfect recall,
there is a disconnected decomposition
$y g_0(Y_0) + \bar{y} g_1(Y_1) + g_2(Y_2)$ such that $g_0, g_1, g_2$
have perfect recall. Lemma \ref{lem:relating-disconnected-decomp}
gives the three possible relations between the two decompositions
$x f_0 + \bar{x} f_1 + f_2$ and $y g_0 + \bar{y} g_1 + g_2$. For cases
(1) and (2), we make use of
Corollary~\ref{cor:equiv-polynomials-perfect} to conclude the
proposition. For case (3), we have $x f_0 + \bar{x} f_1 \equiv g_2$
and $f_2 \equiv y g_0 + \bar{y} g_1$. It is easy to see that $f_2$ has
perfect recall since $g_0$ and $g_1$ have perfect recall. Let
$f' = xf_0 + \bar{x} f_1$. We know that $f'$ has perfect recall, has
fewer variables than $f$ and $x f_0 + \bar{x} f_1$ is a disconnected
decomposition of $f'$. By induction hypothesis, $f_0$ and $f_1$ have
perfect recall. \qed

Note proposition~\ref{prop:perfect-poly-disconnected-decomp} claims that ``every'' disconnected
decomposition is a witness to perfect recall. This way the question of
detecting perfect recall boils down to finding disconnected
decompositions recursively.

\subsubsection*{Finding disconnected decompositions.}

The final step is to find disconnected decompositions.  Given a
polynomial $f$ and $b \in \{0, 1\}$, we say \emph{$x$ cancels $y$ with
  $b$} if substituting $x = b$ in $f$ results in a polynomial without
$y$-terms (neither $y$ nor $\bar{y}$ appears after the
substitution). For a set of variables $S$, we say $x$ cancels $S$ with
$b$ if it cancels each variable in $S$ with $b$. We say that $x$
cancels $y$ if it cancels it with either $0$ or $1$.

\begin{lemma}
  \label{lem:cancellation-not-both-0-and-1}
  Let $f$ be an arbitrary polynomial and $x, y$ be variables. Variable
  $x$ cannot cancel $y$ with both $0$ and $1$ in $f$.
\end{lemma}

\begin{lemma}
  Let $f \in M(X)$ and let $g$ be the polynomial obtained by rewriting
  every $\bar{x}t$ by $t - tx$. Then, $x$ cancels $y$ with $b$ in $f$
  iff $x$ cancels $y$ with $b$ in $g$.
\end{lemma}

\begin{corollary}
  For $b \in \{0, 1\}$ and $x,y \in X$, we have $x$ cancels $y$ with
  $b$ in $f$ iff $x$ cancels $y$ with $b$ in the full expansion of
  $f$.
\end{corollary}

\begin{corollary}
  \label{cor:equiv-polynomials-same-cancellation}
  Let $f, g$ be polynomials such that $f \equiv g$. Then for
  $b \in \{0, 1\}$ and $x \in X$, we have
  $\{ y \mid x \text{ cancels } y \text{ with } b \text{ in } f \}$
  equal to
  $\{ y \mid x \text{ cancels } y \text{ with } b \text{ in } g \}$
\end{corollary}

\begin{lemma}
  Let $x f_0(X_0) + \bar{x} f_1(X_1) + f_2(X_2)$ be an
  $x$-decomposition of $f$. Then, the decomposition is disconnected
  iff for $b \in \{0,1\}$, $X_b$ equals $\{y \mid x \text{ cancels $y$
    with $b$ in $f$} \}$.
\end{lemma}
\begin{proof}
  Suppose $x f_0(X_0) + \bar{x} f_1(X_1) + f_2(X_2)$ is
  disconnected. Then clearly, the conclusion to the forward
  implication follows. Now consider an $x$-decomposition
  $x f_0(X_0) + \bar{x} f_1(X_1) + f_2(X_2)$ which is not necessarily
  disconnected to start off with. Call this decomposition $g$. Suppose
  $X_b = \{y \mid x \text{ cancels $y$ with $b$ in $f$} \}$. By
  definition $g \equiv f$. From
  Corollary~\ref{cor:equiv-polynomials-same-cancellation}, $f$ and $g$
  have the same cancellations due to $x$. Here we make a claim that a
  variable $x$ cannot cancel $y$ with both $0$ and $1$. This claim can
  be easily shown. This shows that $X_0 \cap X_1 = \emptyset$. We know
  that $x \notin X_2$ by definition of the $x$-decomposition. If some
  $y \in X_0 \cap X_2$ then $y$ cannot get canceled by $x$ with
  respect to $0$ in $g$ and hence also in $f$. This shows that
  $X_0 \cap X_2 = \emptyset$. Similar argument also shows that
  $X_1 \cap X_2 = \emptyset$.
\end{proof}

This lemma provides a mechanism to form disconnected
$x$-decompositions starting from a polynomial $f$, just by finding
variables that get canceled and then grouping the corresponding
terms.

\begin{theorem}
  \label{thm:ptime-algo-perfect-poly}
  There is a polynomial-time algorithm to detect if a polynomial has
  perfect recall.
\end{theorem}
\begin{proof}
  Here is the (recursive) procedure.
  \begin{enumerate}
  \item Iterate over all variables to find a variable $x$ such that
    the $x$-decomposition $x f_0(X_0) + \bar{x} f_1(X_1) + f_2(X_2)$
    of $f$ is disconnected. If no such variable exists, stop and
    return \emph{No}.
  \item Run the procedure on $f_0, f_1$ and $f_2$.
  \item Return \emph{Yes}.
  \end{enumerate}
  When the algorithm returns \emph{Yes}, the decomposition witnessing
  the perfect recall can be computed. When the algorithm returns
  \emph{No}, it means that the decomposition performed in some order
  could not be continued. However
  Proposition~\ref{prop:perfect-poly-disconnected-decomp} then says
  that the polynomial cannot have perfect recall.
\end{proof}

The combination of Theorems~\ref{thm:perfect-recall-poly-game} and
\ref{thm:ptime-algo-perfect-poly} gives a heuristic for polynomial
optimization: check if it is perfect recall, if yes convert it into a
game and solve it, if not perform the general algorithm that is
available. This heuristic can also be useful for imperfect recall
games. The payoff polynomial of an imperfect recall game could as well
be perfect recall (based on the values of the payoffs). Such a
structure is not visible syntactically in the game whereas the
polynomial reveals it. When this happens, one could solve an
equivalent perfect recall game.

\section{Pure strategies and bridge}
\label{sec:bounding-randomness}

We have seen that maxmin
computation over behavioural strategies is as hard as solving very generic
optimization problems of multivariate
polynomials over reals. Here we investigate the case of pure
strategies. We first recall the status of the problem.

\begin{theorem}\cite{KollerMegiddo::1992}
  The question of deciding if maxmin value over pure strategies is at
  least a given rational is $\Sigma_2$-complete in two player imperfect
  recall games. It is $\NP$-complete when there is a single player.
\end{theorem}

In this section we refine this complexity result in two ways:
we introduce the chance degree of a game
and show polynomial-time complexity when the chance
degree is fixed; 
next we provide a focus on a tractable class of games called bidding games,
suitable for the study of Bridge.

\subsection{Games with bounded chance}

We investigate a class of games where the $\chance$
player has restrictions. In many natural games, the number of
$\chance$ moves and the number of options for $\chance$ are limited -
for example, in Bridge there is only one $\chance$ move at the very
beginning leading to a distribution of hands. With this intuition, we
define a quantity called the \emph{chance degree} of a game.

\begin{definition}[Chance degree]
  For each node $u$ in the game, the chance degree
  $\cdeg(u)$ is defined as follows: $\cdeg(u) = 1$ if $u$ is a leaf,
  $\cdeg(u) = \sum_{u \to v} \cdeg(v)$ if $u$ is a chance node, and
  $\cdeg(u) = \max_{u \to v} \cdeg(v)$ if $u$ is a control node. The
  chance degree of a game is $\cdeg(r)$ where $r$ is the root. 
\end{definition}

The chance degree in essence expresses the number of leaves reached with
positive probability when players play only pure strategies. For example,
the chance degrees of games $G_2$ (Figure~\ref{fig:game-examples}) and
$G_{-\sqrt{n}}$ (Figure~\ref{fig:sqrt-n}) are $1$ and $2$
respectively.

\begin{lemma}\label{lem:fixed-chance-degree-one-player}
  Let $G$ be a one player game with imperfect recall,  
  chance degree $K$ and $n$ nodes.
    When both players play pure strategies, the number of leaves reached
  is atmost $K$.
   The optimum value over pure strategies can be computed in time
  $\Oo(n^K)$.
\end{lemma}
\begin{proof}
The first statement follows from an induction on the number of non-terminal nodes.

  Partition the set of leaves into bags so that leaves arising out of
  different actions from a common $\chance$ node are placed in
  different bags. Here is an algorithm which iterates over each leaf
  starting from the leftmost till the rightmost, and puts it in a
  corresponding bag. Suppose the algorithm has visited $i$ leaves and
  has distributed them into $j$ bags. For the next leaf $u$, the
  algorithm finds the first bag where there is no $v$ such that the
  longest common prefix in $\pathto(u)$ and $\pathto(v)$ ends with a
  $\chance$ node. If there is no such bag, a new bag is created with
  $u$ in it. It can be shown that the number of bags created is equal
  to the chance degree $K$ of the game.

  In the partitioning above, for every $\chance$ node $u$ and for
  every pair of transitions $u \xra{a} u_1$ and $u \xra{b} u_2$, the
  leaves in the subtrees of $u_1$ and $u_2$ fall in different
  bags. Moreover two leaves differ only due to control nodes and hence
  while playing pure strategies, both these nodes cannot both be reached
  with positive probability.  Therefore, once this partition is created, a pure
  strategy of the player can be seen as a tuple of leaves
  $\langle u_1, \dots, u_m \rangle$ with at most one leaf from each
  bag such that for every stochastic node $u$ which is an ancestor of
  some $u_i$, there is a leaf $u_j$ in the subtree (bag) of every
  child of $u$. The payoff of the strategy is given by the sum of
  $\Cc(t) \Uu(t)$ for each leaf $t$ in the tuple where $\Uu(t)$ is the
  payoff and $\Cc(t)$ is the chance probability to reach $t$. This
  enumeration can be done in $\Oo(n^K)$.
\end{proof}

\begin{theorem}
  \label{thm:fixed-chance-degree}
  Consider games with chance degree bounded by a constant $K$. Optimum
  in the one player case can be computed in polynomial-time. In the
  two player case, deciding if maxmin is at least a rational $\lambda$
  is $\NP$-complete.
\end{theorem}
\begin{proof}
  Lemma~\ref{lem:fixed-chance-degree-one-player} says that the optimum
  for a single player can be computed in $\Oo(n^K)$ where $n$ is the
  number of nodes. Since $K$ is fixed, this gives us
  polynomial-time. For the two player case, note that whenever $\Max$
  fixes a strategy $\s$, the resulting game is a one player game in
  which $\Min$ can find its optimum in polynomial-time. This gives the
  $\NP$ upper bound. The $\NP$-hardness follows from Proposition 2.6
  of \cite{KollerMegiddo::1992} where the hardness gadget has no
  $\chance$ nodes. Hence hardness remains even if chance degree is $1$.
\end{proof}

Since the two player decision problem is hard even when fixing the
chance degree, we need to look for strong structural restrictions
that can give us tractable algorithms. Perfect recall is of course one
of them. In the subsequent section, we consider a model of the bidding
phase of bridge as an imperfect recall game, and investigate some
abstraction that can guarantee polynomial-time.

 \subsection{A model for Bridge bidding} 
\label{sec:model-bridge-bidding}

We propose a model for the Bridge bidding
phase. We first describe the rules of a game which abstracts the
bidding phase. Then we represent it as a zero-sum extensive
form imperfect recall game.

\subsubsection*{The bidding game. } There are four players $N,S,W,E$ in
this game model, representing the players North, South, West and East
in Bridge. Players $N,S$ are in team $T_{max}$ and $E,W$ are in team
$T_{min}$. For a player $i \in \{N, S, W, E\}$, we write $T_i$ to
denote the team of player $i$ and $T_{\neg i}$ for the other
team. This is a zero-sum game played between teams $T_{max}$ and
$T_{min}$. Every player has the same set of actions $\{0,\dots,n\}$
where $0$ imitates a pass in Bridge and action $j$ signifies that a
player has bid $j$. Each player $i$ has a set $H_i$ of possible
private signals (also called secrets). Let $H =H_N \times H_E \times H_S \times
H_W$. Initially each player $i$ receives a private signal from $H_i$
following a probabilistic distribution $\Delta(H)$ (in Bridge, this
would be the initial hand of cards for each player). The game is
turn-based starting with $N$ and followed by $E,S,W$ and proceeds in
the same order at each round. Each player can play a bid which is
either $0$ or strictly greater than the last played non-zero bid. The
game ends when i) $N$ starts with bid 0 and each of $E,S,W$ also
follow with bid $0$ or ii) at any point, three players consecutively
bid $0$ or iii) some player bids $n$. At the end of the game the last
player to have played a non-zero bid $k$ is called the declarer, with
contract $k$ equal to this bid. It is $0$ if everyone bids 0
initially. The payoff depends on a set of given functions
$\Theta_i: H \mapsto \{0, \dots, m\}$ with $m \le n$ for each player
$i$. The function $\Theta_i(\langle h_N, h_E, h_S, h_W \rangle)$ gives
the optimal bid for player $i$ as a declarer based on the initial
private signal $h$ received. The payoff for the teams $T_{max}$ and
$T_{min}$ are now computed as follows: when $i$ is the declarer with
contract $k$ and $h \in H$ is the initial private signal for $i$, if
$\Theta_i(h) \geq k$, $T_i$ gets payoff $k$ whereas $T_{\neg i}$ gets
$-k$. If $\Theta_i(h) < k$ , $T_i$ gets $-k$ and $T_{\neg i}$ gets
$k$.

As an example of this model consider a game where
$H_E = H_W = \{\bot\}$ and $H_N = H_S =
\{\spadesuit,\diamondsuit\}$. There are four possible combinations of
signals in $H$, and the players receive each of them with probability
$\frac{1}{4}$. Players $E,W$ have trivial private signals known to all
and so $\Theta$ does not depend on their signal. A $\Theta$ function
for $n = 5, m = 4$ is given in Figure~\ref{tab:bridge-game-example}. For
example, when the initial private signal combination is
$(\spadesuit, \bot, \spadesuit, \bot)$ and $N$ is the declarer, then the contract
has to be compared with $4$. For the same secret, if $S$ is the
declarer then the contract has to be compared with $2$. 
The longest possible bid sequence in this game is
$(0,0,0,1,0,0,2,0,0,3,0,0,4,0,0,5)$. Let us demonstrate team payoffs
with a few examples of bid sequences. For the initial private signals
$(\spadesuit,\bot,\spadesuit,\bot)$ and the bid sequence
$(0,1,0,2,4,0,0,0)$, $N$ is the declarer with contract $4$, and
$T_{max}$ and $T_{min}$ get payoff $4$ and $-4$ respectively. On
private signals $(\spadesuit,\bot,\diamondsuit,\bot)$ and the bid
sequence $(2,3,0,0,0)$, $E$ is the declarer with contract $3$ and
$T_{max}$ and $T_{min}$ receive payoffs $3$ and $-3$ respectively.

\subsubsection*{Bidding games in extensive form.}
Given a bidding game with the specifications as mentioned above, we
can build an extensive form game corresponding to it. The root node is
a $\chance$ node with children $H$ and transitions giving $\Delta(H)$.
All the other nodes are control nodes. We consider them to belong to
one of the four players $N, E, S, W$. However finally we will view it
as a zero-sum game played between $T_{max}$ and $T_{min}$. These
intermediate nodes are characterized by sequences of bids leading to
the current state of the play. Let $Seq$ be the set of all possible
sequences of bids from $\{0,\dots,n\}$ due to game play. The set $Seq$
also contains the empty sequence $\epsilon$. The nodes in the
extensive form game are the elements of $Seq$. For each sequence $s$
there is a set of valid next moves which contain $0$ and the bids
strictly bigger than the last non-zero bid in $s$. These are the
actions out of $s$. Leaves are bid sequences which signal the end of
the play. The utility at each leaf is given by the payoff received by
$T_{max}$ at the end of the associated bid sequence.

Finally, we need to give the information sets for each player.  Let
$Seq_i$ be the sequences that end at a node of player $i$.  Each
player observes the bid of other players and is able to distinguish
between two distinct sequences of bids at his turn. But, player $i$
does not know the initial private signals received by the other
players. Hence the same sequence of bids from a secret of $i$ and each
combination of secrets of the other players falls under one
information set. More precisely, let
$\mathcal{H}_i = H_i \times Seq_i$ be the set of histories of player
$i$. Two nodes of player $i$ are in the same information set if they
have the same history in $\mathcal{H}_i$. Note that each individual
player $N, E, W, S$ has perfect recall. When considered as a team,
$T_{max}$ and $T_{min}$ have imperfect recall. The initial signal for
a team is a pair of secrets $(h_N, h_S)$ or $(h_E, h_W)$ and within an
information set of say $N$, there are nodes $u$ and $v$ coming from
different initial signals $(h_N, h_S)$ and $(h_N, h'_S)$. This makes
the game a signal-loss recall for each team. Therefore the only
general upper bound for maxmin computation is $\efr$ with behavioural
strategies and $\Sigma_2$ with pure strategies. Observe that the
chance degree of the game is $|H|$ since there is a single $\chance$
node.  When we bound this initial number of secrets $H$ by some $K$,
and vary the bids and payoff functions, we get a family of games with
bounded chance degree. Theorem~\ref{thm:fixed-chance-degree} gives
slightly better bounds for computing the maxmin over pure strategies for this
family of games, which is still $\NP$-hard for the two-player
case. This motivates us to restrict the kind of strategies considered
in the maxmin computation. We make one such attempt below.

\begin{figure}[t]
  \centering
  \begin{tabular}{|c|c|c|c|c|}
    \cline{2-5}
    \multicolumn{1}{c|}{} & \multicolumn{4}{c|}{$\Theta$} \\
    \hline
    Player & $(\diamondsuit,\diamondsuit)$ & $(\diamondsuit,\spadesuit)$& $(\spadesuit,  \diamondsuit )$ & $(\spadesuit , \spadesuit )$ \\
    \hline
    $N$ & 0 & 0& 2 & 4\\
    \hline
    $E$ & 0& 0 & 0 &  0 \\
    \hline
    $S$ & 0& 2 &0 & 2 \\
    \hline 
    $W$ & 0&0 & 0 & 0\\
    \hline 
  \end{tabular}
  \caption{Example of a bidding game}
  \label{tab:bridge-game-example}
\end{figure}

\begin{figure}[t]
  \begin{tabular}{|c|c|c|c|c|c|c|}
    \cline{2-7}
    \multicolumn{1}{c|}{} & \multicolumn{6}{c|}{$\Theta$} \\
    \hline
    Player & $h_1$ & $h_2$ & $h_3$ & $h_4$  & $h_5$ &$h_6$\\
    \hline
    $N$ & 3 & 4& 5  & 0 &0 &0 \\
    \hline
    $E$ & 1& 3 & 2 &  2 & 2&4\\
    \hline
    $S$ & 0& 0 &0 & 3 & 4 & 5\\
    \hline 
    $W$ & 0&0 & 0 & 0 & 0&0\\
    \hline 
  \end{tabular}
  \caption{A second example of a bidding game}
  \label{tab:bridge-game-example-2}
\end{figure}

\subsubsection*{Non-overbidding strategies.}

A pure strategy for player $i$ is a function
$\sigma_i: \mathcal{H}_i \mapsto \{0,\dots, n\}$. In the example of
Figure~\ref{tab:bridge-game-example}, $N$ has to pass on the
information whether she has $\diamondsuit$ or $\spadesuit$ to $S$, and
in the case that $N$ has $\spadesuit$, player $S$ has to pass back
information whether she has $\diamondsuit$ or $\spadesuit$ so that in
the latter case $N$ can bid for $4$ in the next turn. When $E$ knows
the strategy of $N$, she can try to reduce their payoff by playing $3$
when $N$ plays $2$ (if she bids $4$, her team loses and $T_{max}$ gets
a payoff $4$ anyway) and not let $S$ over-bid to pass information to
$N$. But in the process $E$ ends up overbidding when $S$ has
$\diamondsuit$ and it makes no difference to the total expected payoff. This gives strategies $\s_N(\diamondsuit) = 0$,
$\s_N(\spadesuit) = 2$, $\s_S(\diamondsuit, 0 b_E) = 0$,
$\s_S(\spadesuit, 0 b_E) = 2$ (when possible),
$\s_S(\spadesuit, 2 0) = 3$ and $\s_S(\spadesuit, 2 3) = 0$, where
$b_E$ is a placeholder for some bid of $E$. When it comes back to $N$
for the second turn and $S$ had played $3$, then $N$ plays $4$ if she
can, otherwise she passes. This pair of strategies achieves the maxmin
payoff.

A pure strategy $\sigma_i$ of player $i$ is said to be
\emph{non-overbidding} if starting from her second turn, player $i$
always bids $0$: more precisely, for $h \in H_i$ and $s \in Seq_i$,
$\s_i(h, s) = 0$ whenever there exists $s_0 \in Seq_i$ with $s_0$ a
proper prefix of $s$. Otherwise, the strategy is said to be
\emph{over-bidding}. The strategy of $N$ above is over-bidding since
$N$ could potentially bid $4$ after $2$. The number of non-overbidding
strategies is $|H_N| \cdot (n+1)$ for $N$ and $|H_S| \cdot (n+1)$ for
$S$ and hence for team $T_{max}$ there are
$|H_N| \cdot |H_S| \cdot (n+1)^2$ non-overbidding
strategies. Similarly there are $|H_E| \cdot |H_W| \cdot (n+1)^2$
non-overbidding strategies for $T_{min}$. These numbers are
drastically smaller compared to the number of pure strategies, which
is exponential in the size of the extensive form (and doubly
exponential in the size of the input description).

\begin{lemma}\label{lemma:nonoverbidding}
  Maxmin value over non-overbidding strategies can be computed in time
  $|H| \cdot (n + 1)^4$.
\end{lemma}

Of course, non-overbidding strategies will not be in general the same
as maxmin over pure. In particular, for the example of
Table~\ref{tab:bridge-game-example} the strategy $\sigma_N$ mentioned
before is over-bidding. It turns out that in some cases, considering
non-overbidding strategies is sufficient. Consider the game given in
Figure~\ref{tab:bridge-game-example-2}. The only player to receive a
private signal is $N$. All others have a publicly known trivial
signal $\bot$. Player $N$ can receive one of $6$ secrets
$h_1, \dots, h_6$. In this case $N$ bids $3,4,5$ from $h_1,h_2,h_3$
making the optimal contract in her first turn. From $h_4,h_5,h_6$ she
bids $0,1,2$ in her first turn and $S$ gaining complete information
about secret of $N$ due to her distinct actions, bids $3,4,5$
respectively if $E$ has not already made those bids. Here
non-overbidding strategies are sufficient to obtain maxmin expected
payoff.

We have exhibited a class of strategies that can be efficiently
computed and which are sufficient for some games. We leave the more
general question of checking how close the value computed by
non-overbidding strategies is to the actual maxmin as part of future
work.

\bibliography{main} 

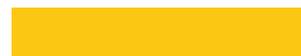
\end{document}